\documentclass[letterpaper,USenglish]{lipics-v2016}

\usepackage{enumitem}
\usepackage[all]{nowidow}
\usepackage{thmtools}
\usepackage{thm-restate}
\usepackage{microtype}

\declaretheorem[name=Theorem,sibling=theorem]{re-thm}
\declaretheorem[name=Lemma,sibling=theorem]{re-lem}

\newcommand{\newparentheses}[3]{%
  \expandafter\newcommand\csname #1\endcsname[1]{#2##1#3}%
  \expandafter\newcommand\csname #1L\endcsname[1]{\bigl#2##1\bigr#3}%
  \expandafter\newcommand\csname #1XL\endcsname[1]{\Bigl#2##1\Bigr#3}%
  \expandafter\newcommand\csname #1XXL\endcsname[1]{\biggl#2##1\biggr#3}%
  \expandafter\newcommand\csname #1V\endcsname[1]{\left#2##1\right#3}}

\newparentheses{parens}{(}{)}
\newparentheses{floor}{\lfloor}{\rfloor}
\newparentheses{ceil}{\lceil}{\rceil}
\newparentheses{abs}{|}{|}
\newparentheses{set}{\{}{\}}
\newparentheses{size}{|}{|}

\makeatletter
\newcommand{\onenewattribute}[3]{%
  \@ifundefined{#1}{\let\@@def\newcommand}{\let\@@def\renewcommand}%
  \expandafter\@@def\csname #1\endcsname[2][]{%
    \ifthenelse{\equal{##1}{}}%
    {#2\csname #3\endcsname{##2}}%
    {#2_{##1}\csname #3\endcsname{##2}}}}
\newcommand{\newattribute}[2]{%
  \onenewattribute{#1}{#2}{parens}%
  \onenewattribute{#1L}{#2}{parensL}%
  \onenewattribute{#1XL}{#2}{parensXL}%
  \onenewattribute{#1V}{#2}{parensV}}
\makeatother

\newattribute{OhOf}{\mathrm{O}}
\newattribute{ThetaOf}{\Theta}
\newattribute{OmegaOf}{\Omega}
\newattribute{ohOf}{\mathrm{o}}
\newattribute{omegaOf}{\omega}
\newattribute{lca}{\text{LCA}}

\newattribute{dspr}{d_{\mathrm{SPR}}}
\newattribute{drspr}{d_{\mathrm{rSPR}}}
\newattribute{duspr}{d_{\mathrm{uSPR}}}
\newattribute{dnni}{d_{\mathrm{NNI}}}
\newattribute{dtbr}{d_{\mathrm{TBR}}}
\newparentheses{degree}{|N(}{)|}

\newcommand{\newproblem}[2]{%
\begin{description}
\item[#1] #2
\end{description}
}

\newcommand{\defproblem}[1]{%
\textbf{#1}%
}

\newcommand{\function}[1]{%
\textsc{#1}%
}

\newcommand{\chomp}[1]{}

\newcommand{\arxiv}[1]{#1}
\newcommand{\notarxiv}[1]{}

\titlerunning{Efficiently Inferring Pairwise SPR Adjacencies between Phylogenetic Trees}
\title{Efficiently Inferring Pairwise Subtree Prune-and-Regraft Adjacencies between Phylogenetic Trees\footnote{This work was funded by National Science Foundation award 1223057 and 1564137. Chris Whidden is a Simons Foundation Fellow of the Life Sciences Research Foundation. The research of Frederick Matsen was supported in part by a Faculty Scholar grant from the Howard Hughes Medical Institute and the Simons Foundation.}}

\author[1]{Chris Whidden}
\author[2]{Frederick A. Matsen IV}
\affil[1]{Program in Computational Biology, Fred Hutchinson Cancer Research Center, Seattle, WA, USA\\
\texttt{cwhidden@fredhutch.org}}
\affil[2]{Program in Computational Biology, Fred Hutchinson Cancer Research Center, Seattle, WA, USA\\
\texttt{matsen@fredhutch.org}}

\authorrunning{C. Whidden and F.\,A. Matsen IV}

\Copyright{Chris Whidden and Frederick A. Matsen IV}



\begin{document}

\maketitle


\begin{abstract}

We develop a time-optimal $\OhOf{mn^2}$-time algorithm to construct the subtree prune-regraft (SPR) graph on a collection of $m$ phylogenetic trees with $n$ leaves.
This improves on the previous bound of $\OhOf{mn^3}$.
Such graphs are used to better understand the behaviour of phylogenetic methods and recommend parameter choices and diagnostic criteria.
The limiting factor in these analyses has been the difficulty in constructing such graphs for large numbers of trees.
We also develop the first efficient algorithms for constructing the nearest-neighbor interchange (NNI) and tree bisection-and-reconnection (TBR) graphs.


\end{abstract}

\section{Introduction}


Phylogenetic methods find an optimal evolutionary tree or a posterior distribution on trees by repeatedly modifying a current tree through a series of ``moves.''
The most commonly applied moves are subtree prune-and-regraft (SPR) moves~\cite{hein1990reconstructing} (Fig.~\ref{fig:spr}) and nearest neighbor interchange (NNI) moves, which are a subset of the SPR moves~\cite{hohna2012guided}.
Some methods also apply tree bisection-and-reconnection (TBR) moves, which are equivalent to applying two SPR moves.
Maximization methods aim to find the ``best'' tree according to an optimization criteria such as likelihood~\cite{Price2010-fi,Stamatakis2006-yz} or parsimony~\cite{swofford2001paup}, while Bayesian statistical methods~\cite{Ronquist2012-hi,bouckaert2014beast} aim to efficiently sample trees.
In both cases the topology of the trees is the most difficult parameter to optimize or sample~\cite{lakner2008efficiency,hohna2012guided,whidden2015quantifying}.
Applying tree-modifying moves in the process of maximization or sampling can be thought of as traversing the graph consisting of trees as vertices and moves as edges.

\begin{figure}[t]
	\centering
	\hspace*{\stretch{2}}
	\subcaptionbox{\label{fig:tree}}{\includegraphics[scale=0.7,trim=0 -8 0 0]{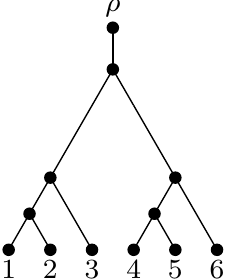}}
	\hspace*{\stretch{1}}
	\subcaptionbox{\label{fig:subtree}}{\includegraphics[scale=0.7,trim=0 -8 0 0]{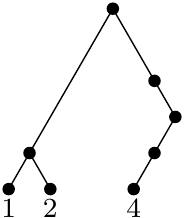}}
	\hspace*{\stretch{1}}
	\subcaptionbox{\label{fig:induced}}{\includegraphics[scale=0.7,trim=0 -8 0 0]{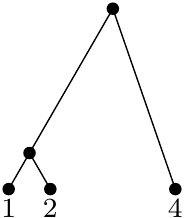}}
	\hspace*{\stretch{1}}
	\subcaptionbox{\label{fig:spr}}{\includegraphics[scale=0.7]{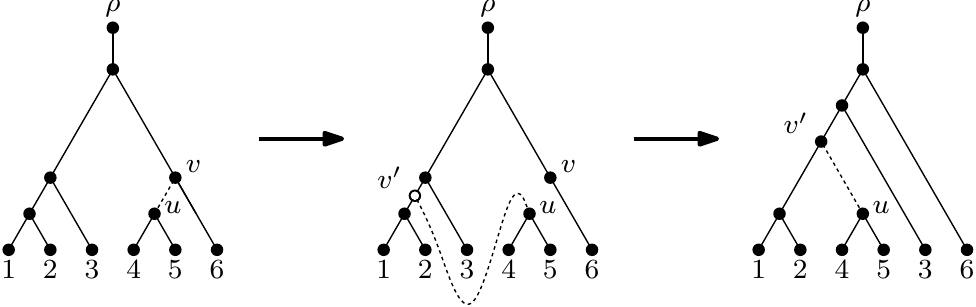}}
	\hspace*{\stretch{2}}

	\caption{(a) A rooted $X$-tree $T$.
		(b) $T(V)$, where $V = \set{1,2,4}$.
		(c) $T|V$.
		(d) An SPR operation transforms $T$ into a new tree by \emph{pruning a subtree} and \emph{regrafting} it in another location.
	}
	\label{fig:trees}
\end{figure}

\begin{figure}[t]
	\hspace*{\stretch{1}}
	\subcaptionbox{\label{fig:ds6}}{\includegraphics[width=0.5\textwidth]{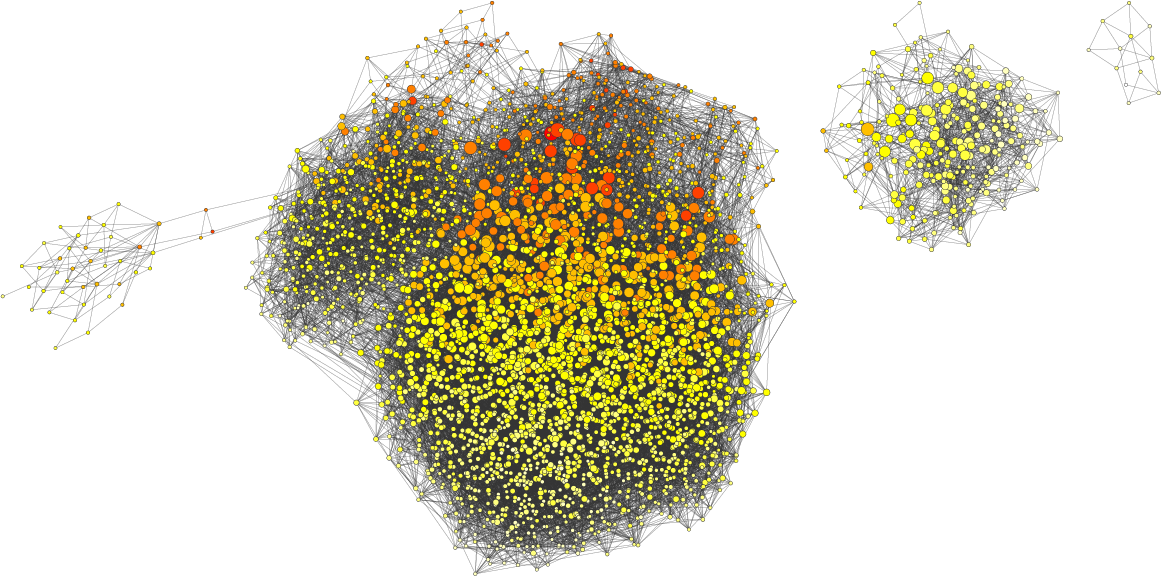}}
	\hspace*{\stretch{2}}
	\subcaptionbox{\label{fig:ds7}}{\includegraphics[width=0.25\textwidth]{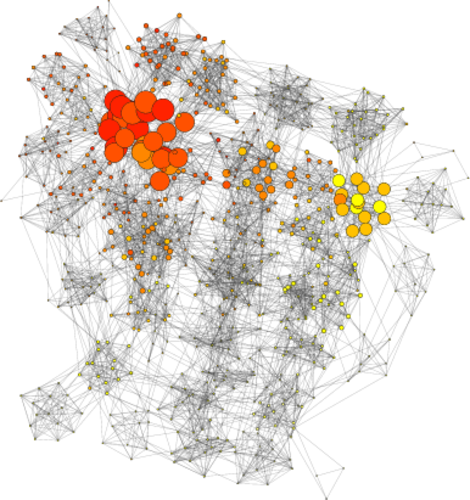}}
	\hspace*{\stretch{1}}

	\caption{Two SPR graph figures of high-probability tree posterior subsets from~\cite{whidden2015quantifying}.
    Node size indicates posterior probability.
    Color (red-yellow-white) indicates SPR distance from the highest probability tree.
    (a) ``Peaky'' distributions separate high probability trees into components.
    (b) Closely related sequences induce lattice-like features.
    }
    \label{fig:spr_graph_examples}
\end{figure}

One can gain insight into the operation of phylogenetic inference methods by explicitly constructing the subgraph composed of trees that have been visited by running an inference method (Fig.~\ref{fig:spr_graph_examples}).
In a highly cited 1991 paper, Maddison~\cite{maddison1991discovery} developed the notion of ``islands'' of neighboring equally-parsimonious trees, and found such islands containing hundreds of trees when running on real data, and indicated their importance for parsimony tree search.
In previous work, we built the subgraph of the SPR graph consisting of the thousands of highest posterior probability trees as inferred by the Markov chain Monte Carlo (MCMC) algorithm, which is the standard means of inferring a posterior distribution on phylogenetic trees.
By doing so, we found significant graph structure relevant for the design of phylogenetic inference software.
Specifically, we found multiple peaks (Fig.~\ref{fig:ds6}), indicating multimodal posteriors, and lattice-like structures (Fig.~\ref{fig:ds7}), indicating a need to collapse closely-related~sequences.

Although graphs connecting a set of phylogenetic trees have been an object of study since 1991 \cite{maddison1991discovery}, the construction of these graphs has not been formulated as a problem for research.
For these early studies, no special methods were needed to build graphs on tens to hundreds of trees.
However modern phylogenetic posterior samples, with hundreds of thousands of trees, demand efficient algorithms.
Indeed, we were limited in our previous work to graphs of several thousand topologies by the lack of efficient algorithms.
This is no trivial task, as it is NP-hard to even determine the minimum distance between a pair of trees in terms of NNI~\cite{dasgupta97computing}, SPR~\cite{bordewich2005computational,hickey2008sdc}, or TBR~\cite{allen01} operations.
Thus we propose:

\newproblem{SPR Graph Construction Problem}{Given $m$ binary phylogenetic trees with $n$ leaves, determine which pairs of trees differ by exactly one SPR move.}

\chomp{Similarly, we can define the corresponding \defproblem{NNI Graph Construction Problem} and \defproblem{TBR Graph Construction Problem} using NNI or TBR moves instead of SPR moves.}

Two methods have previously been introduced for constructing SPR graphs, and we are not aware of any previous methods for constructing NNI or TBR graphs.
The first method~\cite{whidden2015quantifying} compares each pair of trees in a collection using a fixed-parameter algorithm~\cite{whidden2013fixed} to determine whether their SPR distance is 1.
Although the SPR distance is NP-hard, this fixed-parameter algorithm scales exponentially only with the distance computed and linearly with $n$.
This pairwise comparison method thus takes $\OhOf{n}$-time for each pair of trees, for a total of $\OhOf{m^2n}$-time ($\OhOf{m^2n^3}$-time for unrooted trees).
Still, pairwise comparisons are only feasible for small SPR graphs, because of the rapidly growing $m^2$ factor.

The second method for constructing SPR graphs~\cite{whidden2016ricci} relies on the observation that SPR graphs are relatively sparse.
Each tree has $\OhOf{n^2}$ SPR neighbors~\cite{Song2003-gf}.
By storing the $\OhOf{n}$-size Newick~\cite{felsenstein1990newick} strings of trees, one can enumerate the neighbors of a given tree in $\OhOf{n^3}$-time.
This neighbor-enumeration method takes $\OhOf{mn^3}$-time to construct an SPR graph of $m$ trees with $n$ leaves.

The biggest obstacle that slowed these methods was the requirement to explicitly consider each possible pair of neighbors.
The pairwise comparison method does so by considering every pair of trees, at the cost of an extra $\OhOf{m}$ factor.
The neighbor-enumeration method directly considers every neighbor of each tree, adding an extra $\OhOf{n}$ factor per tree for Newick string operations.
All these methods consider trees as the objects and look for connections between them in the SPR graph using structures called agreement forests (AFs).

In this paper we use agreement forests as the objects of interest, which we enumerate and store using new algorithms and data structures. We contribute:
\begin{itemize}
\item A time-optimal $\OhOf{mn^2}$-time algorithm for the (rooted and unrooted) \defproblem{SPR Graph Construction Problem}
\item An $\OhOf{mn^2}$-time algorithm for the \defproblem{NNI Graph Construction Problem}
\item A time-optimal $\OhOf{mn^3}$-time algorithm for the \defproblem{TBR Graph Construction Problem}.
\end{itemize}

The SPR and TBR algorithms are optimal in the sense that their running times correspond to the number of possible edges in the corresponding graphs given $n$ and $m$.
The algorithms are enabled by a variant of the Newick string format, dubbed \emph{smallest descendant label Newick (SDLNewick)}, that can uniquely represent agreement forests, and
a new \emph{AFContainer} data structure that stores and compares tree adjacencies using SDLNewick strings of AFs.
We have deferred proofs besides that of our main result to an appendix\notarxiv{ in the full arXiv version of this paper~\cite{whidden2017efficiently}}, as well as the TBR and NNI graph algorithms.

\section{Preliminaries}

\begin{figure}[t]
	\hspace*{\stretch{1}}
	\subcaptionbox{\label{fig:adjacency_triangle}}{\includegraphics[scale=0.74]{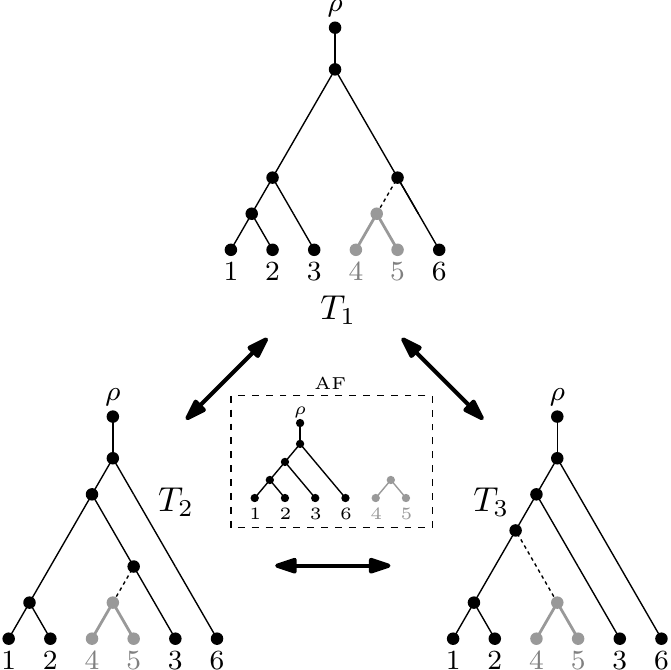}}
	\hspace*{\stretch{2}}
	\subcaptionbox{\label{fig:spr_graph}}{\includegraphics[width=0.35\textwidth]{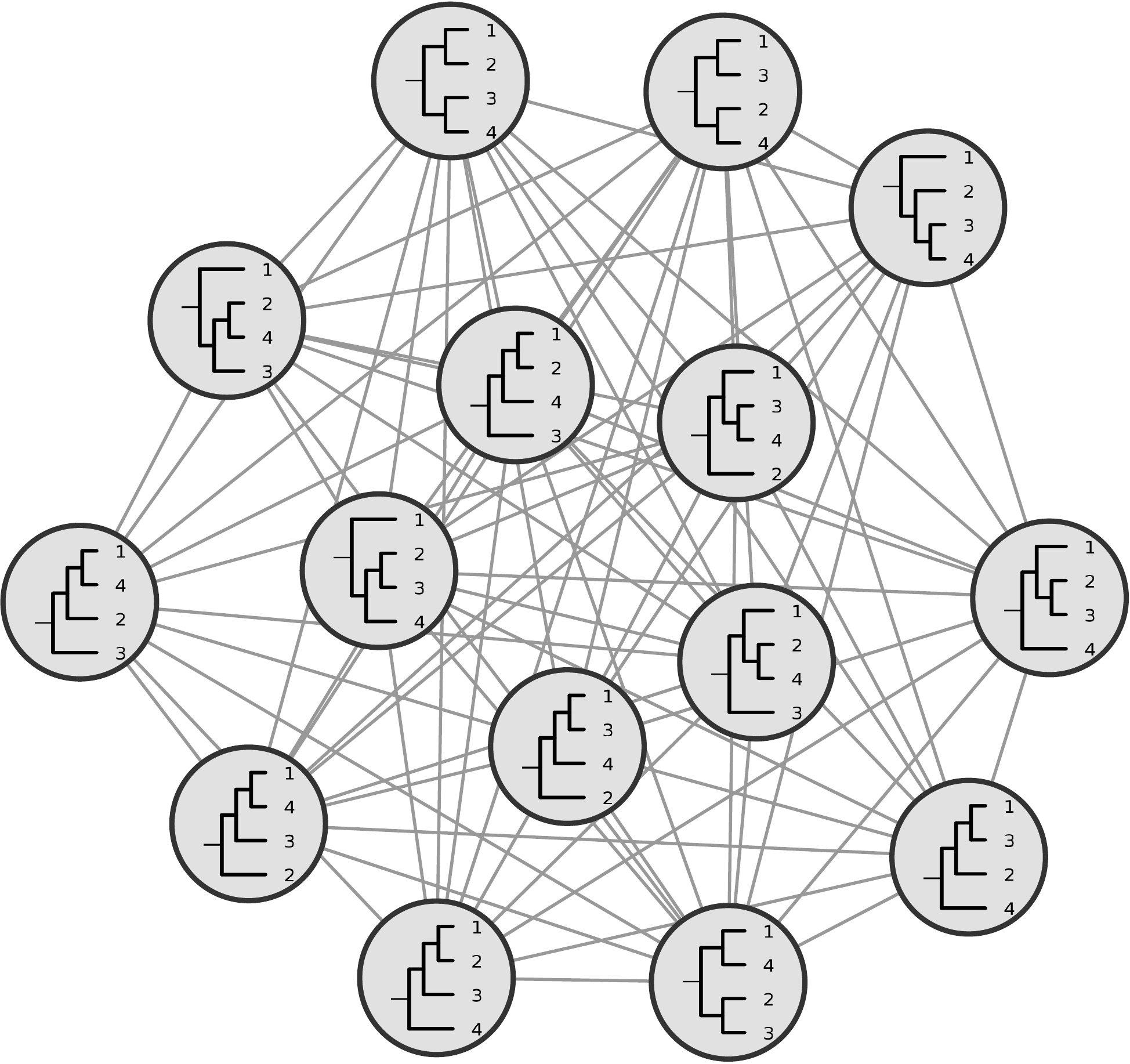}}
	\hspace*{\stretch{1}}

	\caption{(a) Three rooted trees that share a 2-component agreement forest (AF).
    	Each can be obtained from the others by an SPR operation moving the subtree induced by leaves 4 and 5.
		(b) The complete SPR graph on 4-leaf rooted trees.
	}
\end{figure}

A tree is an acyclic graph.
The \emph{leaves} of a tree are nodes with one neighbor and \emph{internal nodes} have multiple neighbors.
An (unrooted binary phylogenetic) $X$-tree is a tree $T$ whose nodes each have one or three neighbors, and whose leaves are bijectively labeled with the members of a label set $X$.
\emph{Suppressing} a node $v$ deletes $v$ and its incident edges; if $v$ has exactly two neighbors $u$ and $w$, then they are reconnected by a new edge $(u,w)$.
$T(V)$ is the unique subtree of $T$ with the fewest nodes that connects all nodes in $V \subset X$.
The $V$-tree induced by $T$ is the smallest tree $T|V$ that can be obtained from $T(V)$ by suppressing unlabeled nodes with fewer than three neighbors.

A \emph{rooted $X$-tree} is defined similarly to an unrooted $X$-tree, with the exception that one of the internal nodes is called the \emph{root} and is adjacent to a leaf labeled $\rho$.
Note that this differs from the standard definition of a rooted tree, in which the root is the only degree two internal node.
This $\rho$ node represents the position of the original root in a forest of the trees, as described below.
Observe that the $\rho$ node can be attached to such a degree two internal node, so these two notions of rooted trees are equivalent.
The \emph{parent} of a node in a rooted tree is its closest neighbor to the root; the other two neighbors are its \emph{children} (Fig.~\ref{fig:trees}).

We assume without loss of generality that the label set $X$ consists of distinct integer values from $1, 2, \ldots n$.
Moreover, for this paper we assume that $n \le 2^{64}-1$ (i.e. able to fit in a standard 64 bit unsigned integer format).
Larger trees are not feasible to infer computationally or logistically.

An \emph{unrooted $X$-forest $F$} is a collection of (not necessarily binary) trees $T_1, T_2, \ldots T_k$ with respective label sets $X_1, X_2, \ldots X_k$.
The label sets are disjoint and complete, that is, $X_i$ and $X_j$ are disjoint, for all $1 \le i \ne j \le k$, and $X = X_1 \cup X_2 \cup \ldots \cup X_k$.
We say $F$ \emph{yields} the forest with components $T_1|X_1, T_2|X_2, \ldots, T_k|X_k$, that is, the smallest forest that can be obtained from $F$ by suppressing unlabeled nodes with fewer than three neighbors.
In the rooted case $\rho \in X_1$ and the unlabeled \emph{component roots} (the nodes that were connected to the $\rho$ component by an edge before cutting) are not suppressed in the yielded forest.
Each component $T_i$ is then rooted at its respective component root.
Only the root of $T_1$ is adjacent to $\rho$ and the remaining roots are of degree two.
For an edge set $E$, $F - E$ denotes the forest obtained by deleting the edges in $E$ from $F$ and $F \div E$ the yielded forest.
For simplicity we say that $F \div E$ is a \emph{forest of} $F$.

A \emph{subtree-prune-regraft} (uSPR) operation~\cite{hein1990reconstructing} on an unrooted $X$-tree $T$ cuts an edge $e = (u,v)$.
This divides $T$ into subtrees $T_u$ and $T_v$, containing $u$ and $v$ respectively.
Then it introduces a new node $v'$ into $T_v$ by subdividing an edge of $T_v$, and adds an edge $(u,v')$.
Finally, $v$ is suppressed.
An rSPR operation is defined similarly on a rooted tree but $v$ must be the parent of $u$ (Fig~\ref{fig:spr}).
If $v'$ is adjacent to $\rho$ then it becomes the root.

A \emph{tree-bisection-and-reconnection} (TBR) operation~\cite{allen01} is similar to a uSPR operation, with the exception that it also introduces a new node $u'$ into $T_u$ by subdividing an edge of $T_u$, adds the edge $(u',v')$ instead of $(u,v)$, and suppresses $u$.
A \emph{nearest-neighbor-interchange} (NNI) operation is an SPR operation where $v$ and the introduced node $v'$ share a neighbor.

SPR operations give rise to a distance measure $\dspr{\cdot,\cdot}$ between $X$-trees defined as the minimum number of SPR operations required to transform one tree into the other.
We distinguish between $\drspr{\cdot,\cdot}$ on rooted trees and $\duspr{\cdot,\cdot}$ on unrooted trees.
The TBR distance $\dtbr{\cdot,\cdot}$ on unrooted trees is defined analogously with respect to TBR operations.
Observe that these distances are the shortest path distances in the respective graphs.

Given trees $T_1$ and $T_2$, a forest $F$ is an \emph{agreement forest} (AF) of $T_1$ and $T_2$ if it is a forest of both trees.
$F$ is a \emph{maximum agreement forest} (MAF) if it has the smallest possible number of components, denoted $m(T_1, T_2)$.
For two unrooted trees $T_1$ and $T_2$, Allen and Steel~\cite{allen01} showed that $\dtbr{T_1, T_2} = m(T_1, T_2) - 1$.
This implies that two unrooted trees which differ by a single TBR operation must share a two component unrooted MAF:

\begin{restatable}{re-lem}{tbr_two_component}
\label{lem:tbr_two_component}
Let $T_1$ and $T_2$ be two distinct unrooted trees.
Then there exists an MAF $F$ of $T_1$ and $T_2$ with two components if, and only if, $\dtbr{T_1, T_2} = 1$.
\end{restatable}

For two rooted trees $T_1$ and $T_2$, Bordewich and Semple~\cite{bordewich2005computational} showed that $\drspr{T_1,T_2} = m(T_1, T_2) - 1$, by introducing the root node augmentation $\rho$ described above.
This implies that two rooted trees which differ by a single SPR operation must share a two component rooted MAF (Fig.~\ref{fig:adjacency_triangle}):

\begin{restatable}{re-lem}{rspr_two_component}
\label{lem:rspr_two_component}
Let $T_1$ and $T_2$ be two distinct rooted trees.
Then there exists an MAF $F$ of $T_1$ and $T_2$ with two components if, and only if, $\drspr{T_1, T_2} = 1$.
\end{restatable}

No general MAF formulation has been identified as equivalent to the unrooted SPR distance and there are reasons to believe that a directly analogous formulation does not exist~\cite{whidden2015calculating}.
However, we prove that two unrooted trees differ by exactly one SPR operation if and only if they share an appropriately defined hybrid two-component MAF.
Given unrooted trees $T_1$ and $T_2$, a forest $F$ is a \emph{uSPR 2-agreement forest} if it is a two-component forest of both trees such that one component is a rooted tree and the other is an unrooted tree.
The node connected to the removed edge in both trees is the component root.

\begin{restatable}{re-lem}{usprtwocomponent}
\label{lem:uspr_two_component}
There exists an uSPR 2-agreement forest $F$ of two distinct unrooted trees, $T_1$ and $T_2$, with two components if, and only if, $\drspr{T_1, T_2} =~1$.
\end{restatable}

\section{A time-optimal SPR graph construction algorithm}
\label{sec:algorithm}

In this section we present our $\OhOf{mn^2}$-time algorithm for the SPR Graph Construction Problem, which operates identically for either rooted and unrooted trees.
The cases of NNI and TBR are similar and addressed in the appendix\notarxiv{~\cite{whidden2017efficiently}}.
The basic idea of the algorithm is to use a new data structure, an AFContainer, to efficiently determine the pairwise SPR adjacencies of a collection of trees $\mathcal{T} = T_1, T_2, \ldots, T_m$.
We first \function{Insert} each tree into the AFContainer in turn and add a vertex corresponding to that tree to the graph.
We then apply the \function{SPRNeighbors} function of the AFContainer in turn for each tree to determine which edges to add to the graph.
The algorithm outputs the SPR graph with vertices labeled $i$ for each tree $T_i$ in $\mathcal{T}$.
We refer to labels as \emph{tree IDs}.

As shown in Lemmas~\ref{lem:rspr_two_component} and~~\ref{lem:uspr_two_component}, two distinct trees are adjacent in the SPR graph if and only if there exists a two-component forest that can be obtained by removing a single edge from both trees.
The AFContainer \function{Insert} function stores a string representation of each of the two-component rooted agreement forests corresponding to each inserted tree.
The AFContainer \function{SPRNeighbors} function then determines which of the previously inserted trees share an agreement forest with the given tree.
We define this data structure in Section~\ref{sec:data_structure}.

Our smallest descendant label Newick (SDLNewick) string representation is based on the venerable Newick tree format but has three important differences.
First, the SDLNewick format distinguishes between rooted and unrooted trees.
Second, the SDLNewick format can represent both trees and forests of trees.
Finally, SDLNewick representations of the same tree or forest are guaranteed to be the same, regardless of the left-right ordering of subtrees.
These features are necessary to easily determine whether two trees share a two-component agreement forest.
We define this string format in detail in Section~\ref{sec:string}.

The high-level steps of the algorithm are as follows:

\begin{enumerate}[label={\arabic*}.]
	\item[] \function{Construct-SPR-Graph($\mathcal{T}$)}
	\item Let $A \leftarrow$ \function{CreateAFContainer}().
	\item Let $G$ be an empty graph.
    \item For $i$ in $1$ to $m$: \vspace{-0.2em}
    \begin{enumerate}[nosep]
		\item Add a vertex $i$ to $G$ representing tree $T_i$.
    	\item $A$.\function{Insert}($T_i$).
    \end{enumerate}
    \item For $i$ in $1$ to $m$: \vspace{-0.2em}
	    \begin{enumerate}[nosep]
            \item Let $N \leftarrow$ $A$.\function{SPRNeighbors($T_i$)}.
            \item for each neighbor ID $j \in N$:
            \begin{enumerate}[nosep]
            	\item Add an edge $e = (j,i)$ to $G$.
            \end{enumerate}
		\end{enumerate}
    Return $G$.
\end{enumerate}

A key factor in achieving our time-optimal $\OhOf{n^2}$ running time bound is allowing and accounting for a small amount of sloppiness from the SPRNeighbors function.
First, we allow the function to return the neighbors of the current tree $T_i$ in an arbitrary order with respect to tree IDs.
Also, the function may return a small number of duplicate IDs caused by pairs of trees with the same agreement forest, at most $\OhOf{n}$ in total (as shown in the proof of Lemma~\ref{lem:neighbors} in the appendix\notarxiv{~\cite{whidden2017efficiently}}).
However, we must also be able to add each edge in constant time to achieve optimality.
To do so, our algorithm always adds edges pointing towards the current tree, $T_i$, in the second for loop.
This ensures that all of the edges starting from a given tree are added to the graph in sorted order with respect to their target.
We can thus add each edge to the end of the corresponding edge list in an adjacency list representation in constant time, even though the set of tree neighbors are not in sorted ID order (see the proof of Theorem~\ref{thm:spr_graph} for details).
Moreover, we can easily avoid adding duplicate edges when the SPRNeighbors function returns duplicate tree ID values.
This is a key requirement for avoiding a log factor in the running time of the algorithm to sort the edges and achieving a full linear speedup over previous algorithms for the graph construction problem.

We now show that this algorithm is correct and time-optimal.

\begin{restatable}{re-thm}{sprgraph}
\label{thm:spr_graph}
SPR Graph Construction can be solved in $\OhOf{mn^2}$-time.
\end{restatable}
\begin{proof}
We first prove the running time bound.
We apply the above \function{Construct-SPR-Graph} algorithm to a collection of trees $\mathcal{T} = T_1, T_2, \ldots, T_m$.
We implement the graph as an adjacency list~\cite{clrs}.
We assume that vertices can be added to the graph and edges can be added to the end of a vertex's edge list in amortized $\OhOf{1}$-time.
This is possible if the edge lists are stored as an array of expandable sorted arrays and each of the graph vertices are indexed by tree IDs.

The algorithm first applies the \function{CreateAFContainer} function in constant time by Lemma~\ref{lem:create}.
In the first loop, the algorithm adds a vertex to the graph, and applies the \function{Insert} function once for each of the $m$ trees.
Adding a vertex to the graph takes constant time per tree.
By Lemma~\ref{lem:insert}, each insertion takes $\OhOf{n^2}$-time for a total of $\OhOf{mn^2}$-time.

In the second loop, the algorithm applies the \function{SPRNeighbors} function once for each of the $m$ trees.
By Lemma~\ref{lem:neighbors} this takes $\OhOf{n^2}$-time for each tree $i$, for a total of $\OhOf{mn^2}$-time.
The algorithm also adds an edge $(j,i)$ to the graph for each neighbor of each tree $i$.
Each tree has $\OhOf{n^2}$ SPR neighbors and by Lemma~\ref{lem:neighbors} each list of returned neighbors contains $\OhOf{n}$ duplicate values.
As we now argue, these edges are added to the end of tree $j$'s edge list, taking $\OhOf{1}$-time each for a total of $\OhOf{mn^2}$-time for all applications of the second loop.

An edge $(j, i)$ can only be added to the graph in the $i$th iteration of the for loop, thus an edge $e_i = (u, v_i)$ is added before any edge $e_j = (u, v_j)$ such that $v_i < v_j$.
The fact that $v_i < v_j$ implies that no such $e_j$ is in the graph when $e_i$ is added.
Thus, edges are always added to the end of an edge list, which takes $\OhOf{1}$-time to either add the edge or determine that the edge already exists.
Therefore the algorithm takes $\OhOf{mn^2}$-time.

Now we prove that the algorithm is correct, that is, the returned graph $G$ is exactly the graph of SPR adjacencies of $\mathcal{T}$.
In the first loop, the algorithm applies the \function{Insert} function once for each of the $m$ trees.
By Lemma~\ref{lem:insert}, this implies that the AFContainer contains each tree in $\mathcal{T}$ and their adjacencies.
The algorithm adds a vertex to $G$ for each tree, so the vertex set of $G$ is $\set{1, 2, \ldots, m}$.

In the second loop, the algorithm applies the \function{SPRNeighbors} function once for each of the $m$ trees.
By Lemma~\ref{lem:neighbors} each application returns the set of SPR neighbors of the corresponding tree $i$.
The algorithm then adds an edge $(j,i)$ to the graph for each neighbor of tree $i$.
We have already shown that the edges are added in sorted order to their respective edge lists.
We will now show that $G$ is exactly the SPR graph of $\mathcal{T}$.

First, suppose that the algorithm adds an edge $(x,y)$ between two trees in $\mathcal{T}$ that are not SPR neighbors.
As shown above, this must have occurred in the $y$th iteration of the second for loop.
However, by Lemma~\ref{lem:neighbors}, $T_x$ must be an SPR neighbor of $T_y$, a contradiction.

Second, suppose that the algorithm adds two or more copies of the same edge.
However, the edges are added in sorted order, so this cannot occur.

Finally, suppose that, when the algorithm terminates, $G$ does not contain an edge $(u,v)$ between two trees in $\mathcal{T}$ that are SPR neighbors.
Consider the $v$th iteration of the second for loop.
By Lemma~\ref{lem:neighbors} and the fact that $u$ and $v$ are SPR neighbors, the list of ID numbers returned by $A$.\function{SPRNeighbors($T_i$)} includes $u$.
Then the algorithm must have added edge $(u,v)$, a contradiction.
Therefore the returned graph $G$ is exactly the SPR Graph of $\mathcal{T}$.
\end{proof}

The difference between rooted and unrooted SPR operations is encapsulated in our data structure and SDLNewick string format.
Indeed, given a mixed set of rooted and unrooted trees the returned graph will contain the rooted adjacencies between rooted trees and the unrooted adjacencies between unrooted trees.

\section{A unique string representation for agreement forests}
\label{sec:string}

In this section, we develop an efficient method of uniquely representing agreement forests as a string of characters.
Numerous methods have been proposed to uniquely represent phylogenetic trees (e.g. ~\cite{whidden2016ricci,junier2010newick}), but none for agreement forests.
Moreover, phylogenetic tree string representations have not been examined formally to our knowledge so we do so here.
Our data structure in Section~\ref{sec:data_structure} compares agreement forests using such strings.

The essential properties of our representation for this use are that it must be: (1) space efficient, (2) quick to encode, (3) quick to decode, and (4) unique.

The standard Newick string format~\cite{felsenstein1990newick} for a rooted tree $T$ is defined recursively, starting at the root node $r$ of $T$.
The Newick format string begins with the label of $r$ (if any), followed by an opening parenthesis ``(''.
Each of the Newick strings for the subtrees rooted at $r$'s children are then appended to the string, separated by commas ``,''.
A closing parenthesis ``)'' is appended to the string to indicate that $r$ has no further children.
A complete Newick string is terminated with a semicolon ``;'', no semicolons are used recursively.

An unrooted tree is represented similarly to a rooted tree, by arbitrarily rooting the tree at an internal node.
If the original tree was binary, this results in a trifurcation at the root of the tree.

The Newick string format fulfills the first three essential properties, that is:
\begin{restatable}{re-lem}{newickproperties}
\label{lem:Newick_properties}
\begin{enumerate}[ref={\ref{lem:Newick_properties}.\arabic*}] \item \label{lem:Newick_properties:space} A Newick string of an $n$-leaf binary tree takes $\OhOf{n}$-space.
\item \label{lem:Newick_properties:encode} A Newick string of a binary tree can be encoded in $\OhOf{n}$-time, and
\item \label{lem:Newick_properties:decode} A Newick string can be decoded to its binary tree in $\OhOf{n}$-time.
\end{enumerate}
\end{restatable}

The Newick string format is, however, not unique.
For any given rooted tree $T$, there are many different Newick string representations, one for each reordering of the children in the tree.
For example, the simple two leaf rooted tree with label set $\set{1,2}$ can be represented with both the Newick string ``(1,2);'' and the string ``(2,1);''.
Moreover, unrooted trees have different Newick string representations for each combination of arbitrary rooting choice and child order, and can not be distinguished from rooted trees with a multifurcation at the root.

To ensure a unique string representation of a given binary tree, we add stricter conditions that force a specific Newick string representation.
We call our variant the \emph{smallest descendant label Newick} string or SDLNewick.
In particular, we fix a unique ordering of children for each node of the tree and a unique rooting for an unrooted tree.
One easy to compute unique child ordering is achieved by sorting children by their smallest descendant label (e.g. ~\cite{whidden2016ricci,junier2010newick}).
The smallest descendant label of each node in the tree can be easily computed in $\OhOf{n}$-time by recursively determining the smallest descendant label of each of a node's children and then taking the minimum of those labels.
The nodes of a bounded degree tree (such as a typical binary tree) can then be reordered in $\OhOf{n}$ time.
The Newick string of the reordered tree will then be unique.
For an unrooted tree, we first root the tree at the internal node adjacent to the leaf with smallest label.
We label the root node of a rooted tree $\rho$ to distinguish between rooted and unrooted trees.
We refer to this procedure as \function{SDLNewick($T$)} in Section~\ref{sec:data_structure}.

\begin{restatable}{re-lem}{SDLNewickproperties}
\label{lem:SDLNewick_properties}
\begin{enumerate}[ref={\ref{lem:SDLNewick_properties}.\arabic*}] \item \label{lem:SDLNewick_properties:space} An SDLNewick string of a binary tree takes $\OhOf{n}$-space,
\item \label{lem:SDLNewick_properties:encode} An SDLNewick string of a binary tree can be encoded in $\OhOf{n}$-time,
\item \label{lem:SDLNewick_properties:decode} An SDLNewick string can be decoded to its binary tree in $\OhOf{n}$-time, and
\item \label{lem:SDLNewick_properties:unique} An SDLNewick string of a binary tree is unique.
\end{enumerate}
\end{restatable}

Finally, we extend SDLNewick to uniquely represent agreement forests of binary trees, our main result of this section.
Recall that these forests are obtained by removing an edge from a tree and suppressing the resulting degree two nodes.
If the same agreement forest can be obtained from two different trees then they are adjacent in the SPR graph.

Let $T$ be a binary tree with label set $X$ and let $F$ be a binary forest of $T$ such that $F = T_0, T_1, \ldots, T_k$.
Each component $T_i$ has label set $X_i$ and, as with agreement forests, $X = X_0 \cup X_1 \cup \ldots \cup X_k$, and $X_i \cap X_j = \emptyset$ for all $0 \le i \ne j \le k$.
We order the components of $F$ by their smallest label.
That is, $T_i < T_j$ if, and only if, $\min(X_i) < \min(X_j)$.
If $T$ is a rooted tree, then the label $\rho_0$ representing its root in $F$ is considered to be a label with a value smaller than all of the other leaf labels.
If any of the other components of $F$ are rooted trees then their roots are labeled $\rho$.
This $\rho$ is considered to be a label with value larger than all of the other leaf labels for the component ordering, but still the smallest label for the purpose of rooting that individual component.
We represent $F$ by appending the SDLNewick strings of its component trees separated by spaces, rather than semicolons, and end the string with a single semicolon.
We call the resulting string the SDLNewick representation of a forest.
We refer to this procedure as \function{SDLNewick($F$)} for use in Section~\ref{sec:data_structure}.
We show that this representation fulfills all four of our essential properties.

\begin{restatable}{re-lem}{SDLNewickAFproperties}
\label{lem:SDLNewick_AF_properties}
\begin{enumerate}[ref={\ref{lem:SDLNewick_AF_properties}.\arabic*}]
\item \label{lem:SDLNewick_AF_properties:space} An SDLNewick string representation of a binary forest takes $\OhOf{n}$-space,
\item \label{lem:SDLNewick_AF_properties:encode} An SDLNewick string representation of a binary forest can be encoded in $\OhOf{n}$-time,
\item \label{lem:SDLNewick_AF_properties:decode} An SDLNewick string can be decoded to its binary forest in $\OhOf{n}$-time, and
\item \label{lem:SDLNewick_AF_properties:unique} An SDLNewick string representation of a binary forest is unique.
\end{enumerate}
\end{restatable}

We close this section by stressing that our SDLNewick string representation applies equally to three types of forests relevant to phylogenetic distance metrics.
By Lemmas~\ref{lem:tbr_two_component},~\ref{lem:rspr_two_component}, and~\ref{lem:uspr_two_component}, two trees are adjacent in the (1) TBR, (2) rooted SPR, or (3) unrooted SPR graphs if and only if they share a two-component forest such that (1) neither component is rooted, (2) both components are rooted, or (3) only the moved component is rooted.

%
%

\section{An efficient data structure for agreement forests}
\label{sec:data_structure}

In this section we introduce our \textbf{AFContainer} data structure for storing and comparing SPR tree adjacencies using agreement forests.
Trees inserted into the AFContainer are given successive unique integer ID numbers starting from 0.
An AFContainer consists of three substructures: the forest trie, the ID trie, and the tree array. 
The \textbf{forest trie} is a trie indexed by SDLNewick forest strings.
Each string represents an agreement forest that can be obtained by removing a single edge of some trees inserted into the AFContainer.
Recall that two trees are adjacent in the SPR graph if, and only if, they share a two-component agreement forest.
The forest trie stores lists of the IDs of those trees.
The \textbf{ID trie} is a trie indexed by SDLNewick tree strings that maps tree strings to tree IDs.
The \textbf{tree array} is an expandable array that maps tree IDs to SDLNewick tree strings.


The data structure supports five main operations.
The \textbf{CreateAFContainer()} function creates a new empty AFContainer. This operation initializes the forest trie, the ID trie, and the tree array.
The \textbf{Insert($T$)} function inserts a tree $T$ into the AFContainer, storing all of the agreement forests that can be obtained by removing any single edge of $T$.
The \textbf{SPRNeighbors($T$)} function finds the IDs of each of the neighbors of a tree $T$ that have been inserted into the AFContainer.
The \textbf{ID($T$)} function returns the integer ID of a tree $T$.
The \textbf{SDLNewick($I$)} function returns the SDLNewick string of the tree with ID $I$.
We present pseudocode for these functions and prove their running time and space properties.

We begin with the comparatively simple \function{CreateAFContainer}() function.

\vspace{-0.5em}
\begin{enumerate}[label={\arabic*}.]
	\item[] \function{CreateAFContainer}()
	\item Create an AFContainer $A$.
	\item Initialize $A$.ForestTrie.
	\item Initialize $A$.IDTrie.
    \item Initialize $A$.TreeArray.
    \item Return $A$.
\end{enumerate}
\vspace{-1em}

\begin{restatable}{re-lem}{create}
\label{lem:create}
An empty AFContainer can be created in constant time.
\end{restatable}

We now present pseudocode for the \function{Insert}($T$) function.
We prove that it takes $\OhOf{n^2}$-time amortized and correctly inserts $T$.

\vspace{-0.5em}
\begin{enumerate}[label={\arabic*}.]
	\item[] \function{Insert}($T$)
	\item Let $I$ be the number of trees in $A$.TreeArray.
    \item Let $S \leftarrow $ \function{SDLNewick}($T$).
    \item If $A$.IDTrie[$S$] exists:
   	\begin{enumerate}[nosep]
    	\item Return.
    \end{enumerate}
    \item Let $A$.IDTrie[$S] \leftarrow I$.
    \item Let $A$.TreeArray[$I$] $\leftarrow S$.
	\item For each edge $e$ of $T$:
    \begin{enumerate}[nosep]
	    \item Let $F \leftarrow $\function{SDLNewick}($T \div e$).
	    \item Add $I$ to $A$.ForestTrie[$F$], creating the list if necessary.
    \end{enumerate}
    \item Return.
\end{enumerate}
\vspace{-1em}

We require three conditions of the insert function given a tree $T$.
After \function{Insert}($T$) returns,

\vspace{-0.5em}
\begin{enumerate}
\item $A$.IDTrie[\function{SDLNewick}($T$)] is a unique integer $I$,
\item $A$.TreeArray[$I$] is \function{SDLNewick}($T$), and
\item For each edge $e$ of $T$, $A$.ForestTrie[$F$] is a list that contains $I$ exactly once, where $F = $ \function{SDLNewick}($T \div e$).
\end{enumerate}
\vspace{-1em}

\begin{restatable}{re-lem}{insert}
\label{lem:insert}
A binary tree can be inserted into an AFContainer in amortized $\OhOf{n^2}$ time.
\end{restatable}

We now present pseudocode for the \function{SPRNeighbors}($T$) function.
We prove that it takes $\OhOf{n^2}$-time and correctly returns all of the ID numbers of neighbors of a tree $T$ that have been inserted into the AFContainer.
Note that there are two limitations of this function that enable us to achieve this running time bound.
First, the neighbor ID numbers are not returned in sorted order.
Second, the list of neighbors will include some duplicate ID values, but only at most $\OhOf{n}$ such duplicates.
This occurs because some pairs of trees share two or more two-component agreement forests.
Neither of these limitations affect our use of this function in Section~\ref{sec:algorithm}.
Moreover, note that both of these limitations can be removed with a sorting pass for use in other applications, for a total of $\OhOf{n^2 \log n}$-time.
We discuss this idea further in Section~\ref{sec:conclusions}.

\vspace{-0.5em}
\begin{enumerate}[label={\arabic*}.]
	\item[] \function{SPRNeighbors}($T$)
	\item Let $I \leftarrow -1$.
    \item If $A$.IDTrie[\function{SDLNewick}($T$)] exists:
    \begin{enumerate}[nosep]
		\item Let $I \leftarrow A$.IDTrie[\function{SDLNewick}($T$)].
    \end{enumerate}
    \item Let $L$ be an empty list of integers.
	\item For each edge $e$ of $T$:
    \begin{enumerate}[nosep]
	    \item Let $F \leftarrow $\function{SDLNewick}($T \div e$).
	    \item If the list $A$.ForestTrie[$F$] is nonempty, append its non-$I$ elements to $L$.
    \end{enumerate}
    \item Return $L$.
\end{enumerate}
\vspace{-1em}

\begin{restatable}{re-lem}{neighbors}
\label{lem:neighbors}
The SPR neighbors of a binary tree $T$ that are stored in an AFContainer can be identified in $\OhOf{n^2}$-time with $\OhOf{n}$ duplicates.
\end{restatable}

Note that the \function{SPRNeighbors} function returns a list of tree IDs rather than the SDLNewick strings of neighboring trees.
This is necessary to achieve an $\OhOf{n^2}$ time bound, as the $\OhOf{n}$ size of each such string implies that a list of strings for all $\ThetaOf{n^2}$ neighbors is of size $\ThetaOf{n^3}$.
Our algorithm in Section~\ref{sec:algorithm} thus uses these tree IDs directly.
With a bounded neighborhood size, however, the neighbor strings can be output more efficiently:

\begin{restatable}{re-lem}{newickneighbors}
\label{lem:Newick_neighbors}
A list of the SPR neighbors of a binary tree $T$ that are stored in an AFContainer can be returned in SDLNewick format in $\OhOf{n^2 + Xn}$-time, where $X$ is the number of neighbors.
\end{restatable}

It is often necessary to determine whether a given tree has been inserted into a data structure and, if so, obtain its identifier.
We now present pseduocode for the AFContainer \function{ID} function. We prove that it takes $\OhOf{n}$-time to find the ID of a tree that has been inserted into the AFContainer or determine that a tree was not previously inserted into the AFContainer.

\vspace{-0.5em}
\begin{enumerate}[label={\arabic*}.]
	\item[] \function{ID}($T$)
	\item Let $S \leftarrow $\function{SDLNewick}($T$).
    \item If $A$.IDTrie[$S$] exists:
    \begin{enumerate}[nosep]
    	\item Return the ID $I$.
    \end{enumerate}
    \item Else:
    \begin{enumerate}[nosep]
    	\item Return $-1$.
    \end{enumerate}
\end{enumerate}
\vspace{-1em}

\begin{restatable}{re-lem}{index}
\label{lem:index}
The ID of a binary tree $T$ in SDLNewick format can be found or determined not to be in an AFContainer in $\OhOf{n}$-time.
\end{restatable}

Similarly, we may need to determine which tree corresponds to a given ID.
It takes $\OhOf{n}$-time to return the SDLNewick string of a tree given its ID number.

\vspace{-0.5em}
\begin{enumerate}[label={\arabic*}.]
	\item[] \function{SDLNewick}($I$)
	\item If $A$.TreeArray[$I$] exists:
    \begin{enumerate}[nosep]
    	\item Return the stored SDLNewick string $S$.
    \end{enumerate}
    \item Else:
    \begin{enumerate}[nosep]
    	\item Return the empty string $``''$.
    \end{enumerate}
\end{enumerate}
\vspace{-1em}

\begin{restatable}{re-lem}{newick}
\label{lem_Newick}
The SDLNewick string corresponding to a tree with ID $I$ can be found in $\OhOf{n}$-time.
\end{restatable}
\vspace{-0.5em}

For our final proof of the basic AFContainer operations, we show that the total space required by an AFContainer holding $m$ trees with at most $n$ leaves is $\OhOf{mn^2}$.

\begin{restatable}{re-lem}{afcontainerspace}
\label{lem:afcontainer_space}
An AFContainer holding $m$ trees requires $\OhOf{mn^2}$ space.
\end{restatable}

\section{Conclusions}
\label{sec:conclusions}

We developed the first time-optimal algorithms for the \defproblem{SPR Graph Construction Problem} and \defproblem{TBR Graph Construction Problem}, and the first efficient algorithm for the \defproblem{NNI Graph Construction Problem}, given $m$ phylogenetic trees with $n$ leaves.

The key insight behind these algorithms was storing and manipulating agreement forests of trees rather than the trees themselves.
To do so, we introduced a new SDLNewick string representation for representing agreement forests, and an AFContainer data structure that stores and compares such strings.
SDLNewick strings are efficient to construct and process and allow one to easily determine whether two trees or agreement forests are the same.
Although there have been many such representations for trees, ours is the first that uniquely distinguishes between rooted and unrooted trees and uniquely represents agreement forests.



The AFContainer is the first efficient method of identifying a large number of adjacencies between evolutionary trees.
We wish to stress that the AFContainer data structure can also be used dynamically, for example to update the graph given a stream of trees.

There are several avenues to explore in future work.
First, our data structure does not currently allow for the deletion of trees.
It may be useful to identify and delete trees that are unlikely with respect to the sequence data as a search progresses to reduce the memory required by the AFContainer.
Second, although a major advance, our algorithm for constructing NNI graphs is not time-optimal as trees have only $\OhOf{n}$ NNI-neighbors.
Closing this gap is an open problem.
Third, although our SPR and TBR graph algorithms are time-optimal in the sense that they match the maximum size of each graph given $n$ and $m$, they do not necessarily match the size of a given graph.
Developing an output-sensitive algorithm which runs in time proportional to the actual size of the constructed graph is a challenge and would be very useful for testing and developing new phylogenetic methods.
Finally, it remains to implement our data structure and apply it to the testing and development of current and new phylogenetic methods.

\bibliographystyle{plainurl}
\bibliography{main}

\newpage

\arxiv{

\appendix

\section{Omitted Proofs}

\usprtwocomponent*
\begin{proof}
Let $T_1$ and $T_2$ be two distinct unrooted trees.

First suppose that there exists such a forest $F = T_1 \div E_1 = T_2 \div E_2$.
Then $E_1$ contains a single edge $e_1$ and $E_2$ contains a single edge $e_2$.
Consider the two components $t_1$ and $t_2$ of $F$, such that $t_1$ is the unrooted component and $t_2$ the rooted component.
Then $e_1$ and $e_2$ are attached to the same node of $t_2$ in both $T_1$ and $T_2$.
We can thus denote the edges $e_1 = (u, v)$ and $e_2 = (u, v')$.
Let $y$ and $z$ be the other neighbors of $v'$ in $T_2$.
Then we can transform $T_1$ into $T_2$ by applying the SPR operation that cuts $e_1$ in $T_1$, introduces the node $v'$ on the edge $(y,z)$ and then connects $u$ and $v'$.

Now, suppose that there exists an SPR operation that transforms $T_1$ into $T_2$ by cutting an edge $(u,v)$, introducing a node $v'$ and adding the edge $(u, v')$.
Then the forest of $T_1$ with rooted component $T_u$ and unrooted component $T_v$ is a forest of $T_1$ and $T_2$ and thus a uSPR 2-agreement forest of $T_1$ and $T_2$.
\end{proof}

\newickproperties*
\begin{proof}
These properties are well known but we are not aware of any proofs that have appeared in scholarly work so we briefly argue their correctness here.
We first consider property (1).
A rooted binary tree with $n$ leaves has $n+1$ internal nodes, each with two children.
By the recursive Newick definition, each internal node adds 3 characters to the format, an opening parenthesis, comma, and closing parenthesis.
Each leaf node adds its label which, by our assumptions on reasonable $n$ takes at most 20 characters.
Finally, the string is terminated by 1 semicolon character.
The Newick representation of an $n$ leaf string thereby consists of at most $3 (n+1) + 20n + 1 = 23n + 4$ characters.

For property (2), we observe that it takes constant time to apply the definition recursively to each node of the tree, so the Newick string can be encoded in linear time.
Similarly, for property (3), a tree can be constructed in linear time by recursively processing a Newick string with a well known algorithm.
Briefly, this consists of creating a new node for each opening parenthesis as a child to the previous node, labeling leaf nodes with the integer labels, returning to the previous parent node when reaching a comma or closing parenthesis, and terminating this procedure when the semicolon is reached.
\end{proof}

\SDLNewickproperties*
\begin{proof}
Let $T$ be a binary tree.
We first observe that an SDLNewick string representation of $T$ is a valid Newick string, as it is the Newick string representation of some reordering of $T$'s edges.
Thus, property (1) follows from Lemma~\ref{lem:Newick_properties:space} and property (3) follows from Lemma~\ref{lem:Newick_properties:decode}.

We next consider property (2).
Let $T'$ be the smallest descendant label reordering of $T$.
By Lemma~\ref{lem:Newick_properties:encode}, we can encode $T'$ to the SDLNewick string representation of $T$ in $\OhOf{n}$-time.
We now show that we can construct $T'$ from $T$ in $\OhOf{n}$-time.
If $T$ is unrooted then we first compute the smallest label of $T$.
This takes $\OhOf{n}$-time to traverse $T$, applying a constant number of operations to each node.
We then reroot $T$ at the internal node adjacent to that leaf.
This also takes $\OhOf{n}$-time to set the root node and then traverse the tree, setting parent pointers from each node.

The final step in constructing $T'$ is determining the child edge reordering and reordering the children.
To do so, we apply a recursive post-order traversal starting at the root of $T$ that (1) determines the smallest descendant label of a node by taking the minimum of the smallest descendant label of each of its children, (2) determines the new child ordering by comparing their smallest descendant labels and (3) reorders the children.
This process applies a constant number of operations per node of the tree.
Thus, $T'$ can be constructed in $\OhOf{n}$-time, and property (2) holds.

Last we show that property (4) holds.
In particular, we show that the above procedure is fully deterministic, that is, two applications will result in the same SDLNewick string given any starting child order of a tree $T$.
By our assumption on tree labels, every leaf has a distinct label.
Thus, there is a unique smallest label, and therefore a unique smallest label rooting if $T$ is unrooted.
The above procedure identifies this unique rooting and applies it.
Moreover, a label cannot be the descendant of two nodes with the same parent, so every node of the tree with the same parent must have a unique smallest descendant label.
Thus, each node of the tree has a unique smallest descendant label child ordering.
It is easy to see using induction that the above procedure identifies this unique ordering and applies it.
Therefore the SDLNewick string representation of $T$ is unique.
\end{proof}

\SDLNewickAFproperties*
\begin{proof}
Let $T$ be a binary tree and $F = T_0, T_1, \ldots, T_k$ be a binary forest of $T$.
Let $S$ be an SDLNewick string representation of $F$.

We first prove property (1).
As noted above, $S$ is the concatenation of SDLNewick strings of each component of $F$ which have been permuted, with each semicolon but the last replaced by space characters.
In other words, $S = $``$\function{SDLNewick}(T_{\pi_0})\ \function{SDLNewick}(T_{\pi_1})\ \ldots\ \allowbreak \function{SDLNewick}(T_{\pi_k})$;'' where $\Pi = \pi_0, \pi_1, \ldots, \pi_k$ is a permutation of the component numbers.
By Lemma~\ref{lem:Newick_properties:space}, $S$ is of size $\OhOf{\size{X_{\pi_0}}} + \OhOf{\size{X_{\pi_1}}} + \ldots + \OhOf{\size{X_{\pi_k}}} = \OhOf{\size{X_0}} + \OhOf{\size{X_1}} + \ldots + \OhOf{\size{X_k}} = \OhOf{n}$.
Thus, the first claim holds.

We now prove property (2).
Let $F'$ be the smallest descendant label reordering of the components of $F$.
By applying Lemma~\ref{lem:SDLNewick_properties:encode} to each component of $F'$, we can encode $F'$ to $S$.
This takes $\OhOf{\size{X_{\pi_0}}} + \OhOf{\size{X_{\pi_1}}} + \ldots + \OhOf{\size{X_{\pi_k}}} = \OhOf{n}$-time.
It thus suffices to show that we can construct $F'$ from $F$ in $\OhOf{n}$-time to prove property (2).
We traverse each component of $F$ in order to determine its smallest label, storing the results in an array of size $k + 1 = \OhOf{n}$.
This takes $\OhOf{\size{X_0}} + \OhOf{\size{X_1}} + \ldots + \OhOf{\size{X_k}} = \OhOf{n}$-time.
We then apply CountingSort~\cite{clrs} to sort the components by their smallest label in $\OhOf{k + n} = \OhOf{n}$-time.

Property (3) follows from the structure of $S$ in a similar fashion to property (1).
We convert each of the space characters in $S$ to semicolons and apply Lemma~\ref{lem:SDLNewick_properties:decode} to each Newick string to construct a binary forest from $S$.
This takes $\OhOf{\size{X_0}} + \OhOf{\size{X_1}} + \ldots + \OhOf{\size{X_k}} = \OhOf{n}$-time.

Finally, we prove property (4).
As in the proof of Lemma~\ref{lem:SDLNewick_properties:unique}, we show that the above encoding procedure is fully deterministic.
By Lemma~\ref{lem:SDLNewick_properties:unique}, the string representation of each component of $F$ is unique.
It thus suffices to show that the component ordering is unique.
By our assumption on tree labels, every leaf has a distinct label with the possible exception of artificial labels $\rho_0$ and $\rho$.
Moreover, only one component can have the $\rho_0$ label that indicates the root of $T$.
Finally, $\rho$ labels have an ordering value larger than any other label.
Thus, each component of the forest has a distinct smallest label, and therefore the smallest label ordering is unique.
\end{proof}

\create*
\begin{proof}
The \function{CreateAFContainer}() function simply initalizes three empty data structures, two tries and an expandable array.
This takes constant time.
\end{proof}

\insert*
\begin{proof}
We first show that each step of the \function{Insert} function other than the for loop can be implemented to take at most $\OhOf{n}$-time amortized over a series of insert operations.
It takes constant time to determine the previous size of the TreeArray, and thereby obtain the new tree index $I$.
By Lemma~\ref{lem:SDLNewick_properties:encode}, it takes $\OhOf{n}$-time to construct the SDLNewick representation $S$ of $T$.
It takes $\OhOf{k}$-time to determine if an entry with key length $k$ exists in a trie.
By Lemma~\ref{lem:SDLNewick_properties:space}, $S$ is of length $\OhOf{n}$.
Thus, it takes $\OhOf{n}$-time to determine if $S$ is already a key in the IDTrie and, if so, terminate the function.
Similarly, it takes $\OhOf{n}$-time to insert $I$ into the IDTrie with key $S$.
$I$ is the next empty element of the TreeArray.
Thus, it takes constant amortized time to insert $S$ into the TreeArray.

Now, consider the for loop.
A tree with $n$ leaves has $\OhOf{n}$ edges, so there are $\OhOf{n}$ iterations of the loop.
The first step of each iteration takes $\OhOf{n}$-time, constant time to remove the edge $e$ and suppress any resulting degree 2 nodes, and linear time to generate the SDLNewick string $F$ (by Lemma~\ref{lem:SDLNewick_AF_properties:encode}).
The second step of each iteration also takes $\OhOf{n}$-time, linear time to determine if the list $A$.ForestTrie[$F$] already exists, linear time to create and insert it into the ForestTrie if it does not, and constant time to add $I$ to the list.
There are a linear number of iterations, each taking at most linear time, so the function can be implemented to take $\OhOf{n^2}$-time.

Finally, we show that the \function{Insert} function is correct, that is, after the function returns, all three correctness conditions hold.
We assume inductively that the conditions hold for any prior tree inserted into the AFContainer.
Recall that the SDLNewick representation of a tree is unique by Lemma~\ref{lem:SDLNewick_properties:unique}.
First, suppose that a tree equivalent to $T$ has been inserted previously.
Then all three conditions already hold prior to applying \function{Insert($T$)}.
Thus $A$.IDTrie[$S$] exists and the function correctly terminates without modifying the AFContainer.

Now, suppose that no tree equivalent to $T$ has been inserted previously.
Then the ID $A$.IDTrie[\function{SDLNewick($T$)}] cannot exist.
Thus the function will assign a new index $I$ to $A$.IDTrie[$S$].
The fact that the chosen value of $I$ is equal to the number of trees in the tree array implies that the index is one larger than any previous index and must be unique.
This fulfills the first condition.
The function then sets $A$.TreeArray[$I$] to \function{SDLNewick($T$)}, fulfilling the second condition.
Finally, suppose that the third condition does not hold when the function terminates.
Then there exists an edge $e$ of $T$ such that the list $A$.ForestTrie[$F$] does not contain $I$ or contains two or more values of $I$, where $F = $ \function{SDLNewick}($T \div e$).
The function considers each edge $e$ of $F$, so the list must contain $I$ at least once.
Furthermore, the function considers each edge exactly once, and no two forests obtained from $T$ by removing different edges can be isomorphic.
Thus, by Lemma~\ref{lem:SDLNewick_AF_properties:unique} no two such forests have the same SDLNewick representation.
Therefore the list contains $I$ exactly once and the function is correct.
\end{proof}

\neighbors*
\begin{proof}

We first show that the algorithm is correct when applied to a binary tree $T$.
Consider the list $L$ returned by \function{SPRNeighbors}($T$).
We will show that $L$ contains the ID numbers of every SPR neighbor of $T$ that has been inserted into the AFContainer and does not contain any other values or more than an $\OhOf{n}$ number of duplicate values.

First, suppose that there exists a tree $T'$ that is an SPR neighbor of $T$ such that $T'$ was inserted into the AFContainer with index $I'$ but $I' \notin L$.
By Lemma~\ref{lem:rspr_two_component} and~\ref{lem:uspr_two_component}, the fact that $T$ and $T'$ are neighbors imply that there exists a forest $F = T \div e = T' \div e'$ where $e$ and $e'$ are edges of $T$ and $T'$, respectively.
Then, by Lemma~\ref{lem:insert}, the list $A$.ForestTrie[\function{SDLNewick}($T' \div e')]$ exists and contains $I'$.
Moreover, we have that \function{SDLNewick}($T' \div e'$) = \function{SDLNewick}($T \div e$) by Lemma~\ref{lem:SDLNewick_AF_properties:unique}.
Then $A$.ForestTrie[\function{SDLNewick}($T \div e)$] exists and contains $I'$.
Thus, \function{SPRNeighbors}($T$) must have appended $I'$ to $L$, a contradiction.

Now, suppose that $L$ contains an integer $I'$ that is not the ID number of an SPR neighbor of $T$ that has been inserted into the AFContainer.
Note that the function only appends non-$I$ values to $L$ from lists in $A$.ForestTrie.
By Lemma~\ref{lem:insert}, $I'$ must be the ID number of a tree $T'$ that has been inserted into the AFContainer.
Moreover, \function{SDLNewick}($T \div e$) must be equal to \function{SDLNewick}($T' \div e'$), where $e$ and $e'$ are edges of $T$ and $T'$, respectively.
However, by Lemma~\ref{lem:rspr_two_component} and~\ref{lem:uspr_two_component}, this implies that $T$ and $T'$ are SPR neighbors, a contradiction.

Finally, suppose that $L$ contains an integer $I$ corresponding to a tree $T'$ two or more times.
We will show that there are $\OhOf{n}$ such duplicate integers.
By Lemma~\ref{lem:insert}, no single list from $A$.ForestTrie contains two or more of the same value.
Then there must exist two distinct forests $T \div e$ and $T \div e'$ such that both lists $A$.ForestTrie[\function{SDLNewick}($T \div e$)] and $A$.ForestTrie[\function{SDLNewick}($T \div e'$)] contain $I$.
That is, $T$ can be transformed into $T'$ by two or more different SPR moves.
Whidden and Matsen~\cite{whidden2016ricci} showed that this occurs if and only if $T$ and $T'$ are also NNI neighbors and that these different moves correspond exactly to NNI moves on $T$.
There are $\OhOf{n}$ NNI moves on $T$.
Therefore there are $\OhOf{n}$ duplicate values in $L$.

We now show that the \function{SPRNeighbors}($T$) function takes $\OhOf{n^2}$-time.
It takes constant time to initialize an empty list.
$T$ has $\OhOf{n}$ edges, so the for loop applies $\OhOf{n}$ iterations.
We will show that each iteration takes linear time, for a total of $\OhOf{n^2}$-time.

It takes linear time to copy $T$ and then constant time to remove $e$ from the copy and suppress degree two nodes in order to construct $T \div e$.
By Lemma~\ref{lem:SDLNewick_AF_properties:encode}, it takes linear time to construct the SDLNewick string $F$ from $T \div e$.
It takes linear time to retrieve a list pointer from a trie with a key of length $\OhOf{n}$.
There are $\OhOf{n}$ trees with the same two-element agreement forest, and no ForestTrie list contains the same tree ID value twice by Lemma~\ref{lem:insert}.
Thus, the list contains $\OhOf{n}$ elements, which are added to $L$ in $\OhOf{n}$-time.
Therefore the running time of the function is $\OhOf{n^2}$ as claimed.
\end{proof}

\newickneighbors*
\begin{proof}
We apply the \function{SPRNeighbors}($T$) function to obtain a list $L$ containing the tree IDs of $T'$ neighbors from the AFContainer.
We then simply apply the AFContainer \function{SDLNewick} function to each ID number to obtain a list $L'$ containing the SDLNewick representations of $T$'s neighbors.
By Lemma~\ref{lem:neighbors}, the first step correctly returns the list of $X$ tree IDs in $\OhOf{n^2}$-time.
By Lemma~\ref{lem:index}, the second step correctly obtains the SDLNewick representations of those trees, using $\OhOf{n}$-time per tree for a total of $\OhOf{Xn}$-time.
Thus, the total time required is $\OhOf{n^2 + Xn}$.
\end{proof}

\index*
\begin{proof}
First, assume that $T$ has been inserted into the AFContainer previously
By Lemma~\ref{lem:insert}, $A$.IDTrie[$S$] contains the ID $I$ of $T$ and the function returns it.
Now, assume that $T$ has not been inserted into the AFContainer previously.
The IDTrie only matches trees with the same SDLNewick string as trees that have been inserted.
Along with Lemma~\ref{lem:SDLNewick_properties:unique}, this implies that the function returns $-1$ indicating that $T$ is not in the AFContainer.

By Lemma~\ref{lem:SDLNewick_properties:encode}, it takes $\OhOf{n}$-time to compute $S$.
By Lemma~\ref{lem:SDLNewick_properties:space}, $S$ is of $\OhOf{n}$ size, so it also takes linear time to look up $S$ in the IDTrie.
All other operations take constant time, so the function takes $\OhOf{n}$-time overall.
\end{proof}

\newick*
\begin{proof}
The correctness of the algorithm follows by similar arguments to those in the proof of Lemma~\ref{lem:index}.
The running time bound follows by noting that it takes constant time to look up an integer-keyed value in an expandable array and $\OhOf{n}$-time to return the (by Lemma~\ref{lem:Newick_properties:space}) $\OhOf{n}$-size string.
\end{proof}

\afcontainerspace*
\begin{proof}
We prove the bound by induction on $m$.
Assume that the claim is true for any number of insertion operations $m' < m$.
Then, after $m - 1$ insertion operations the AFContainer takes $c_0(m-1)n^2$ space, where $c_0 > 0$ is a constant.

Consider the $m$th insertion operation, \function{Insert}($T$).
Let $I$ be the new ID for $T$ and $S = $ \function{SDLNewick}($T$).
We note again that $S$ takes $\OhOf{n}$-size by Lemma~\ref{lem:SDLNewick_properties:space}.
The Insert function increases the space used by the AFContainer in three ways, (1) adding $I$ to the IDTrie with key $S$, (2) adding $S$ to the TreeArray with key $I$, and (3) adding the two-component agreement forests obtained from $T$ to the ForestTrie.
Adding an integer value to a trie with a key of length $\OhOf{n}$ adds $\OhOf{n} \le c_1n$ space, for a constant $c_1 > 0$.
Adding a string value of length $\OhOf{n}$ to an expandable array requires $\OhOf{n} \le c_2n$ space, for a constant $c_2 > 0$.
There are $\OhOf{n}$ edges of $T$ and hence $\OhOf{n}$ updates to the ForestTrie.
By Lemma~\ref{lem:SDLNewick_AF_properties:space},  each new ForestTrie key is of length $\OhOf{n}$.
Therefore these updates cumulatively take $\OhOf{n^2} \le c_3n^2$ space, for a constant $c_3 > 0$.
Then the increase in the space used by the AFContainer is $c_1n + c_2n + c_3n^2$.

Let $c = \max(c_0, 3c_1, 3c_2, 3c_3)$.
The total space used by the AFContainer after applying the $mth$ insertion operation is then $c_0(m-1)n^2 + c_1n + c_2n + c_3n^2 \le c(m-1)n^2 + (c/3)n + (c/3)n + (c/3)n^2 \le c(m-1)n^2 + cn^2 \le cmn^2$.
Therefore the total space used by the AFContainer is $\OhOf{n^2}$.
\end{proof}

\section{Fast algorithms for the NNI and TBR graph construction problems}
\label{sec:algorithm_extended}

In this section we show how to modify our algorithm from Section~\ref{sec:algorithm} to construct NNI and TBR Graphs.
We again have a collection of trees $\mathcal{T} = T_1, T_2, \ldots, T_m$.

We first consider the NNI Graph Construction Problem, and show that it can also be solved in $\OhOf{mn^2}$-time.
The basic idea of the algorithm remains the same, to \function{Insert} each tree into the AFContainer in turn, add a vertex corresponding to that tree to the graph, and then add the edges to the graph.
NNI operations are a subset of SPR operations, so we can use the same \function{Insert} function that we used in the \function{Construct-SPR-Graph} algorithm.
We then apply the \function{NNINeighbors} function (see Section~\ref{sec:data_structure_extended}) in turn to each tree to determine which edges to add to the graph.

The high-level steps are as follows:

\vspace{0.5em}
\function{Construct-NNI-Graph($\mathcal{T}$)}
\vspace{-1em}
\begin{enumerate}[label={\arabic*}.]
	\item Let $A \leftarrow$ \function{CreateAFContainer}().
	\item Let $G$ be an empty graph.
    \item For $i$ in $1$ to $m$:
    \begin{enumerate}[nosep]
		\item Add a vertex $i$ to $G$ representing tree $T_i$.
    	\item $A$.\function{Insert}($T_i$).
    \end{enumerate}
    \item For $i$ in $1$ to $m$:
	    \begin{enumerate}[nosep]
            \item Let $N \leftarrow$ $A$.\function{NNINeighbors($T_i$)}.
            \item for each neighbor ID $n \in N$:
            \begin{enumerate}[nosep]
            	\item Add an edge $e = (n,i)$ to $G$.
            \end{enumerate}
		\end{enumerate}
    Return $G$.
\end{enumerate}

It is now straightforward to show that the algorithm is correct and bounded by our claimed running time.

\begin{restatable}{re-thm}{nnigraph}
\label{thm:nni_graph}
The NNI Graph Construction problem can be solved in $\OhOf{mn^2}$-time.
\end{restatable}
\begin{proof}
The only change in \function{Construct-NNI-Graph} from \function{Construct-SPR-Graph} is the use of the \function{NNINeighbors} function instead of the \function{SPRNeighbors} function.
Therefore, the correctness and running time bound follows from similar arguments to those in the proof of Theorem~\ref{thm:spr_graph} using Lemma~\ref{lem:nni:neighbors} in place of Lemma~\ref{lem:neighbors}.
\end{proof}

Finally, we consider the TBR Graph Construction Problem, and show that it can be solved in $\OhOf{mn^3}$-time.
The additional $\OhOf{n}$ factor in the running time stems from the fact that trees have $\OhOf{n^3}$ TBR neighbors as opposed to $\OhOf{n^2}$ SPR neighbors.
The two main operations are again inserting trees into the AFContainer and identifying tree adjacencies.

TBR operations are a superset of SPR operations, so the same \function{Insert} function that we used in the \function{Construct-SPR-Graph} algorithm  cannot be used here as it does not include information to identify TBR adjacencies.
We instead apply a \function{TBRInsert} function (Section~\ref{sec:data_structure_extended}) that accounts for the fact that TBR adjacencies are uniquely determined by unrooted maximum agreement forests~\cite{allen01} rather than the rooted maximum agreement forests that identify SPR adjacencies.
We then apply the \function{TBRNeighbors} function (also see Section~\ref{sec:data_structure_extended}) in turn to each tree to determine which edges to add to the graph.

The high-level steps are as follows:

\vspace{0.5em}
\function{Construct-TBR-Graph($\mathcal{T}$)}
\vspace{-1em}
\begin{enumerate}[label={\arabic*}.]
	\item Let $A \leftarrow$ \function{CreateAFContainer}().
	\item Let $G$ be an empty graph.
    \item For $i$ in $1$ to $m$:
    \begin{enumerate}[nosep]
		\item Add a vertex $i$ to $G$ representing tree $T_i$.
    	\item $A$.\function{TBRinsert}($T_i$).
    \end{enumerate}
    \item For $i$ in $1$ to $m$:
	    \begin{enumerate}[nosep]
            \item Let $N \leftarrow$ $A$.\function{TBRNeighbors($T_i$)}.
            \item for each neighbor ID $n \in N$:
            \begin{enumerate}[nosep]
            	\item Add an edge $e = (n,i)$ to $G$.
            \end{enumerate}
		\end{enumerate}
    Return $G$.
\end{enumerate}

It is now straightforward to show that the algorithm is correct and bounded by our claimed running time.

\begin{restatable}{re-thm}{tbrgraph}
\label{thm:tbr_graph}
The TBR Graph Construction problem can be solved in $\OhOf{mn^3}$-time.
\end{restatable}
\begin{proof}
The two changes in \function{Construct-TBR-Graph} from \function{Construct-SPR-Graph} are the use of the \function{TBRInsert} function instead of the \function{Insert} function and the use of the \function{TBRNeighbors} function instead of the \function{SPRNeighbors} function.
Therefore, the correctness and running time bound follows from similar arguments to those in the proof of Theorem~\ref{thm:spr_graph} using Lemma~\ref{lem:tbr:insert} in place of Lemma~\ref{lem:insert} and Lemma~\ref{lem:tbr:neighbors} in place of Lemma~\ref{lem:neighbors}.
\end{proof}

\section{An efficient data structure for comparing NNI and TBR agreement forests}
\label{sec:data_structure_extended}

In this section we extend our AFContainer data structure from Section~\ref{sec:data_structure} to infer NNI and TBR adjacencies.
The basic substructures of the AFContainer remain the same.
To infer NNI adjacencies, we rely on the fact that there are only $\OhOf{n}$ NNI neighbors of a given tree with $n$ leaves.
This allows us to directly infer each NNI neighbor at a cost of $\OhOf{n}$-time each while maintaining the same overall quadratic running time of the \function{SPRNeighbors} function.

There are $\OhOf{n^3}$ TBR neighbors of any given tree, however, a linear factor larger than the number of SPR neighbors.
In addition, TBR operations are a superset of SPR operations and we require additional information to infer TBR adjacencies.
We introduce a new \function{TBRInsert} function that stores the two-component unrooted agreement forests which correspond to TBR adjacencies as opposed to the two-component rooted SPR agreement forests or two-component partially unrooted SPR agreement forests.
We then apply a new \function{TBRNeighbors} function that uses these agreement forests to infer the TBR adjacencies, in an analogous manner to \function{SPRNeighbors}.

We first present pseudocode for the \function{NNINeighbors} function.
We assume an arbitrary smallest descendant label rooting if the input tree is unrooted.
In the following, the \emph{aunt edge} of an edge $e$ is the edge that is sibling to $e's$ parent edge.

\vspace{0.5em}
\function{NNINeighbors}($T$)
\vspace{-1em}
\begin{enumerate}[label={\arabic*}.]
    \item Let $L$ be an empty list of integers.
	\item For each edge $e$ of $T$ with an aunt edge:
    \begin{enumerate}[nosep]
	    \item Let $T'$ be the tree obtained by the NNI operation moving the subtree rooted below $e$ to it's aunt edge.
	    \item If $A$.IDTrie[\function{SDLNewick}($T'$)] is nonempty, append its value to $L$.
    \end{enumerate}
    \item Return $L$.
\end{enumerate}

We show that this function can be implemented to take $\OhOf{n^2}$-time.

\begin{restatable}{re-lem}{nnineighbors}
\label{lem:nni:neighbors}
A list of the NNI neighbors of a binary tree $T$ that are stored in an AFContainer can be returned in SDLNewick format in $\OhOf{n^2}$-time.
\end{restatable}
\begin{proof}
Let $T$ be a binary tree with $n$ leaves.

We first show that the \function{NNINeighbors}($T$) function takes $\OhOf{n^2}$-time.
$T$ has $\OhOf{n}$ edges, and so the for loop is applied $\OhOf{n}$ times.
We now show that each loop iteration takes $\OhOf{n}$-time, for a total of $\OhOf{n^2}$-time.
It takes $\OhOf{n}$-time to copy $T$ and apply an NNI operation to obtain $T'$.
By Lemma~\ref{lem:SDLNewick_properties:space} and Lemma~\ref{lem:SDLNewick_properties:encode}, it takes $\OhOf{n}$-time to obtain the SDLNewick string for $T'$.
It then takes $\OhOf{n}$-time for a trie lookup in the IDTrie with that string as the key, and constant time to append an integer to a list, if the lookup is successful.

Now, we show that the algorithm is correct, that is it returns a list containing the ID values of every NNI neighbor of $T$ that has been inserted into the AFContainer and no other values.
Assume that this is not true, for the purpose of obtaining a contradiction.
We first observe that the algorithm only adds values to $L$ from the IDTrie.
By Lemma~\ref{lem:insert}, these correspond to trees that have been inserted into the AFContainer.
Then there are two cases, depending on whether $L$ contains an ID of a tree that is not an NNI neighbor of $T$ or $L$ is missing an NNI neighbor of $T$.
Consider the first case.
Then there exists a tree $T''$ with ID $I \in L$ that is not an NNI neighbor of $T$.
Consider the iteration of the for loop with tree $T'$ that added $I$ to $L$.
By Lemma~\ref{lem:insert}, $T''$ and $T'$ must have the same SDLNewick string.
Then, by Lemma~\ref{lem:SDLNewick_properties:encode}, $T'$ and $T''$ are the same tree.
However, $T'$ was obtained from $T$ by an NNI operation, contradicting the fact that $T''$ is not an NNI neighbor of $T$.

Now consider the second case.
Then there exists a tree $T''$ with ID $I \notin L$ such that $T''$ is an NNI neighbor of $T$.
Whidden and Matsen~\cite{whidden2016ricci} showed that the exact set of NNI neighbors of a tree can be obtained by NNI operations on aunt edges.
Then there exists an edge $e$ of $T$ such that $T''$ can be obtained by the NNI operation moving the subtree rooted below $e$ to it's aunt edge.
Thus, $I$ would have been added to $L$ in the iteration of the for loop that considered $e$, a contradiction.
\end{proof}

We next present pseudocode for the \function{TBRInsert} function.
This function is similar to the \function {Insert} function with the exception that the agreement forest keys of the ForestTrie are unrooted agreement forests.
This is achieved by removing the root label leaf from the second component induced by each edge removal.
Note that a single AFContainer can not be used to infer both SPR and TBR adjacencies, as any second insert function on the same tree is ignored to prevent duplicate IDTrie keys.
However, it would not be difficult to introduce a function that duplicated the behaviour of both the \function{Insert} and \function{TBRInsert} functions by applying both for loops.

\vspace{0.5em}
\function{TBRInsert}($T$)
\vspace{-1em}
\begin{enumerate}[label={\arabic*}.]
	\item Let $I$ be the number of trees in $A$.TreeArray.
    \item Let $S \leftarrow $ \function{SDLNewick}($T$).
    \item If $A$.IDTrie[$S$] exists:
   	\begin{enumerate}[nosep]
		\item Return.
    \end{enumerate}
    \item Let $A$.IDTrie[$S] \leftarrow I$.
    \item Let $A$.TreeArray[$I$] $\leftarrow S$.
	\item For each edge $e$ of $T$:
    \begin{enumerate}[nosep]
    	\item Let $e_\rho$ be the edge adjacent to $\rho$ in $T$.
	    \item Let $F \leftarrow $\function{SDLNewick}($T \div e \div e_\rho \setminus \set{\rho}$).
	    \item Add $I$ to $A$.ForestTrie[$F$], creating the list if necessary.
    \end{enumerate}
    \item Return.
\end{enumerate}

We again require three conditions of the insert function given a tree $T$.
After \function{TBRInsert}($T$) returns,
\begin{enumerate}
\item $A$.IDTrie[\function{SDLNewick}($T$)] is a unique integer $I$,
\item $A$.TreeArray[$I$] is \function{SDLNewick}($T$), and
\item For each edge $e$ of $T$, $A$.ForestTrie[$F$] is a list that contains $I$ exactly once, where $F = $ \function{SDLNewick}($T \div e \div \rho$).
\end{enumerate}

\begin{restatable}{re-lem}{tbrinsert}
\label{lem:tbr:insert}
A binary tree and its unrooted agreement forests can be inserted into an AFContainer in $\OhOf{n^2}$-time.
\end{restatable}
\begin{proof}
The proof follows analogously to that of Lemma~\ref{lem:insert}.
\end{proof}

We now present pseudocode for the \function{TBRNeighbors}($T$) function.
Again, the only difference from the \function{SPRNeighbors} function is the use of unrooted agreement forests.
We prove that it takes $\OhOf{n^3}$-time and correctly returns all of the ID numbers of TBR neighbors of a tree that have been inserted into the AFContainer with the \function{TBRInsert} function.
Again, as with the \function{SPRNeighbors} function, the ID numbers are not sorted and there may be $\OhOf{n}$ duplicate ID values in the list.

\vspace{0.5em}
\function{TBRNeighbors}($T$)
\vspace{-1em}
\begin{enumerate}[label={\arabic*}.]
    \item Let $L$ be an empty list of integers.
	\item Let $I \leftarrow -1$.
    \item If $A$.IDTrie[\function{SDLNewick}($T$)] exists:
    \begin{enumerate}[nosep]
		\item Let $I \leftarrow A$.IDTrie[\function{SDLNewick}($T$)].
    \end{enumerate}
	\item For each edge $e$ of $T$:
    \begin{enumerate}[nosep]
    	\item Let $e_\rho$ be the edge adjacent to $\rho$ in $T$.
	    \item Let $F \leftarrow $\function{SDLNewick}($T \div e \div e_\rho \setminus \set{\rho}$).
	    \item If the list $A$.ForestTrie[$F$] is nonempty, append its non-$I$ elements to $L$.
    \end{enumerate}
    \item Return $L$.
\end{enumerate}

\begin{restatable}{re-lem}{tbrneighbors}
\label{lem:tbr:neighbors}
The TBR neighbors of a binary tree $T$ that are stored in an AFContainer can be identified in $\OhOf{n^3}$-time.
\end{restatable}
\begin{proof}
The proof follows analogously to that of Lemma~\ref{lem:neighbors}, using Lemma~\ref{lem:tbr:insert} in place of Lemma~\ref{lem:insert} and the fact that there are $\OhOf{n^3}$ TBR neighbors of an $n$-leaf tree as opposed to $\OhOf{n^2}$ SPR neighbors.
\end{proof}

Finally, we again consider the case where one wishes to obtain the SDLNewick strings of a set of TBR neighbors, rather than just their ID numbers.
As was the case with SPR, this adds a linear factor to the amount of computation required.
Thus, this approach takes $\OhOf{n^4}$-time in the worst case.
However, this is again not necessarily more computationally expensive to compute when the fraction of the TBR neighborhood stored in the AFContainer is small.

\begin{restatable}{re-lem}{tbrNewickneighbors}
\label{lem:tbr:Newick_neighbors}
A list of the TBR neighbors of a binary tree $T$ that are stored in an AFContainer can be returned in SDLNewick format in $\OhOf{n^3 + Xn}$-time, where $X$ is the number of neighbors.
\end{restatable}
\begin{proof}
The proof follows analogously to that of Lemma~\ref{lem:Newick_neighbors} using Lemma~\ref{lem:tbr:insert} in place of Lemma~\ref{lem:insert}, Lemma~\ref{lem:tbr:neighbors} in place of Lemma~\ref{lem:neighbors}, and Lemma~\ref{lem:tbr_two_component} in place of Lemmas~\ref{lem:rspr_two_component}~and~\ref{lem:uspr_two_component}.
\end{proof}

} 

\end{document}